\documentclass[11pt]{article}
\usepackage[utf8]{inputenc}

\iffalse
\usepackage[accepted]{icml2021_arxiv}
\else
\usepackage{fullpage}
\usepackage{algorithm}
\usepackage{algorithmic}
\fi

\usepackage{xcolor}
\usepackage{amsmath}
\usepackage{amssymb}
\usepackage{amsthm}
\usepackage[]{algorithm}
\usepackage{hyperref}
\usepackage{cleveref}
\usepackage{graphicx,nicefrac}
\usepackage{subcaption}
\usepackage[font={small}]{caption}
\usepackage{thm-restate}
\usepackage[normalem]{ulem}
\usepackage{float}
\usepackage{balance}

\usepackage{enumitem}
\usepackage{xcolor}
\usepackage[framemethod=TikZ]{mdframed}

\newcommand{\nc}{\newcommand}
\newcommand{\DMO}{\DeclareMathOperator}
\DeclareMathAlphabet\mathbfcal{OMS}{cmsy}{b}{n}

\allowdisplaybreaks

\nc{\MS}{\mathcal{S}}
\nc{\MP}{\mathcal{P}}
\nc{\MR}{\mathcal{R}}
\nc{\cM}{\mathcal{M}}
\nc{\cS}{\mathcal{S}}
\nc{\cI}{\mathcal{I}}
\nc{\cA}{\mathcal{A}}
\nc{\tcA}{\tilde{\cA}}

\nc{\MZ}{\mathcal{Z}}

\DMO{\Binom}{Binom}
\newcommand{\E}{\mathbb{E}}
\DMO{\Var}{Var}

\newcommand{\bX}{\mathbf{X}}
\nc{\tbx}{\tilde{\bx}}
\nc{\tbX}{\tilde{\bX}}
\nc{\tZ}{\tilde{Z}}
\nc{\tz}{\tilde{z}}
\newcommand{\bU}{\mathbf{U}}
\nc{\tbU}{\tilde{\bU}}
\newcommand{\bT}{\mathbf{T}}
\nc{\tbT}{\tilde{\bT}}
\newcommand{\bD}{\mathbf{D}}
\nc{\tbD}{\tilde{\bD}}

\newcommand{\bx}{\mathbf{x}}

\newcommand{\R}{\mathbb{R}}
\newcommand{\N}{\mathbb{N}}

\nc{\BN}{\mathbb{N}}

\nc{\BZ}{\mathbb{Z}}

\newcommand{\cE}{\mathcal{E}}

\newcommand{\ind}{\mathbf{1}}
\newcommand{\bA}{\mathbf{A}}
\nc{\tbA}{\tilde{\bA}}

\newcommand{\eps}{\epsilon}

\newcommand{\cval}{\mathrm{cval}}
\newcommand{\val}{\mathrm{val}}
\newcommand{\indp}{\mathrm{indep}}
\newcommand{\td}{\tilde{d}}

\newcommand{\clip}{\mathrm{clip}}

\newtheorem{theorem}{Theorem}

\newtheorem{lemma}[theorem]{Lemma}

\newtheorem{definition}[theorem]{Definition}

\newcommand{\Inf}{\mathsf{Inf}}

\title{Hardness of Approximating Bounded-Degree Max 2-CSP \\ and Independent Set on $k$-Claw-Free Graphs}
\date{July 2023}

\newif\ifarxiv
\arxivtrue

\ifarxiv
\author{ {Euiwoong Lee\footnote{Supported in part by NSF grant CCF-2236669 and Google.}} \\ University of Michigan\\ \texttt{euiwoong@umich.edu} \and {Pasin Manurangsi}\\ Google Research\\ \texttt{pasin@google.com} }
\else
\author{}
\fi

\begin{document}

\maketitle

\begin{abstract}
We consider the question of approximating Max 2-CSP where each variable appears in at most $d$ constraints (but with possibly arbitrarily large alphabet). There is a simple $(\frac{d+1}{2})$-approximation algorithm for the problem. We prove the following results for any sufficiently large $d$:
\begin{itemize}
\item Assuming the Unique Games Conjecture (UGC), it is NP-hard (under randomized reduction) to approximate this problem to within a factor of $\left(\frac{d}{2} - o(d)\right)$.
\item It is NP-hard (under randomized reduction)  to approximate the problem to within a factor of $\left(\frac{d}{3} - o(d)\right)$.
\end{itemize}
Thanks to a known connection~\cite{DvorakFRR23}, we establish the following hardness results for approximating Maximum Independent Set on $k$-claw-free graphs:
\begin{itemize}
\item Assuming the Unique Games Conjecture (UGC), it is NP-hard (under randomized reduction) to approximate this problem to within a factor of $\left(\frac{k}{4} - o(k)\right)$.
\item It is NP-hard (under randomized reduction)  to approximate the problem to within a factor of $\left(\frac{k}{3 + 2\sqrt{2}} - o(k)\right) \geq \left(\frac{k}{5.829} - o(k)\right)$.
\end{itemize}
In comparison, known approximation algorithms achieve $\left(\frac{k}{2} - o(k)\right)$-approximation in polynomial time~\cite{Neuwohner21,ThieryW23} and $(\frac{k}{3} + o(k))$-approximation in quasi-polynomial time~\cite{CyganGM13}.
\end{abstract}


\section{Introduction}

An instance of Max 2-CSP consists of variables--each of which can take a value of an alphabet--together with constraints, each involving a pair of variables. The goal is to find an assignment to the variables that satisfies as many constraints as possible. The Max 2-CSP problem is a cornerstone of the field of hardness of approximation as it\footnote{Or more precisely, its special case known as \emph{Label Cover} or \emph{Projection Games}} is often used as a starting point in hardness of approximation reductions. When the constraints are restricted to certain predicates--such as 3SAT or Max-Cut, tight hardness of approximation results are known through a series of influential work (e.g.~\cite{Hastad01,KhotKMO07}). In fact, it is known that a certain semi-definite program relaxation provides essentially the best approximation ratio achievable in polynomial time~\cite{Raghavendra08}. Meanwhile, besides the predicate, there are also other parameters that can affect the approximation ratio. One of which is the \emph{(maximum) degree} of the instance $d$, defined as the maximum number of constraints that a variable appears in. A number of previous studies have investigated how the degree affects the approximation ratio (e.g.~\cite{Hastad00,Trevisan01,Laekhanukit14,BarakMORRSTVWW15}), partly because, as we will see in more detail below, it affects the (in)approximation ratio of subsequent problems in hardness reductions. In this work, we focus on determining the approximation ratio in terms of this parameter $d$ alone (regardless of the predicate or other parameters) and ask: 
\begin{center}
\emph{What is the best possible approximation ratio for Max 2-CSP in terms of $d$?}
\end{center}
Regarding this question, there is a simple $(\frac{d+1}{2})$-approximation algorithm for the problem (see \Cref{sec:approx-algo}). On the hardness front, Laekhanukit~\cite{Laekhanukit14} proved NP-hardness (under randomized reduction) with a factor of $\Omega(d / \log d)$ for any sufficiently large $d$. Furthermore, under a less standard ``Strongish Planted Clique Hypothesis''\footnote{The Strongish Planted Clique Hypothesis states that no $n^{o(\log n)}$-time algorithm can distinguish between a $G(n, 1/2)$ random graph and one in which a clique of size $n^c$ is planted for some absolute constant $c > 0$.}, a hardness of a factor $\Omega(d)$ is known but without any explicit constant in the inapproximability factor\footnote{In fact, this hardness only states that, for each $c_1 > 0$, there exists $c_2 > 0$ such that no $O(n^{c_1})$-time algorithm achieves $c_2d$-approximation ratio. In other words, it does not rule out e.g. $n^{O(\log^* n)}$ time algorithm from achieving $o(d)$-approximation ratio.}~\cite{ManurangsiRS21}.

\paragraph{Maximum Independent Set on $k$-Claw-Free Graphs.}
While the Maximum Independent Set problem is well known to be NP-hard to approximate to within a factor of $n^{1 - \eps}$ where $n$ is the number of vertices~\cite{Hastad96,Zuckerman07}, there are multiple classes of graphs for which this can be significantly improved upon. One such class is that of \emph{$k$-claw-free} graphs. Recall that a $k$-claw (i.e. $K_{1, k}$) is the star graph with a center vertex connecting to $k$ other vertices (where there are no edges between these $k$ vertices). A graph is $k$-claw-free if it does not contain a $k$-claw as an \emph{induced} subgraph. The classic local search algorithm of Berman~\cite{Berman00} achieves $(\frac{k}{2} + \eps)$-approximation in polynomial time for any constant $\eps > 0$. Recently, this has been improved by~\cite{Neuwohner21,ThieryW23} to achieve a slightly-better-than-$(k/2)$ approximation ratio; in particular, \cite{ThieryW23} achieves approximation ratio of $(\frac{k}{2} - \frac{1}{3} + o(1))$ where $o(1)$ is a term that converges to 0 as $k \to \infty$. Meanwhile, in quasi-polynomial (i.e. $n^{O(\log n)}$) time, this ratio can be improved\footnote{We only discuss the unweighted case in our paper as our hardness results apply to this case; for the weight case, it is not known how to achieve $\left(\frac{k}{3} + o(k) \right)$-approximation in quasi-polynomial time. Please refer to \cite{Neuwohner-arxiv} for the best approximation algorithms known for the weighted case.} to $\left(\frac{k}{3} + \eps\right)$ for any constant $\eps > 0$~\cite{CyganGM13}. These algorithms are based on local search approaches. Meanwhile, several works have also investigated the power of LP/SDP relaxations of the problem: Chudnovsky and Seymour~\cite{CHUDNOVSKY2010560} showed that a standard SDP relaxation yields 2-approximation when $k = 3$, but a recent work by Chalermsook et al.~\cite{chalermsook2023independent} shows large integrality gaps for $k > 3$.

On the hardness of approximation front, Hazan et al.~\cite{HazanSS06} proved that the problem is NP-hard to approximate to within $\Omega(k/\log k)$ factor. In a recent work, Dvorak et al.~\cite{DvorakFRR23} observed that the classic FGLSS reduction~\cite{FeigeGLSS96} provides an approximation-preserving reduction from Max 2-CSP with maximum degree $d$ to maximum independent set in $2d$-claw-free graphs. This reduction, together with the aforementioned hardness from \cite{Laekhanukit14}, gives an alternative NP-hardness proof with a similar inapproximability factor. Meanwhile, plugging this to the other aforesaid result of~\cite{ManurangsiRS21} implies that no polynomial-time algorithm can achieves $o(k)$ approximation ratio (albeit without an explicit constant again) under the Strongish Planted Clique Hypothesis. \cite{DvorakFRR23} also consider parameterization based on the independent set size and prove several hardness results in that setting; we defer the discussion on this to \Cref{sec:fpt}.

\subsection{Our Results}

\paragraph{Bounded-Degree Max 2-CSP.}

Our main contribution is a nearly tight hardness of approximation result for Max 2-CSP in terms of $d$ assuming the Unique Games Conjecture (UGC)~\cite{Khot02}:

\begin{theorem} \label{thm:ug-2csp-main}
Assuming the UGC, for any $\eps \in (0, 1/2)$, there exists $d_0 \in \N$ such that the following holds for every positive integer $d \geq d_0$: Unless NP = BPP, there is no polynomial-time $d(1/2 - \eps)$-approximation algorithm for $d$-bounded-degree Max 2-CSP.  
\end{theorem}

As stated earlier, there is a simple $(\frac{d+1}{2})$-approximation algorithm and thus our result is within a factor of $1 + o(1)$ of this upper bound (as $d \to \infty$). To the best of our knowledge, this is also the first $\Omega(d)$ hardness of approximation result with an explicit constant for the problem (under any assumption). For NP-hardness, we prove a slightly weaker result where the factor is instead $\approx d/3$:

\begin{theorem} \label{thm:np-2csp-main}
For any $\eps \in (0, 1/3)$, there exists $d_0 \in \N$ such that the following holds for every positive integer $d \geq d_0$: Unless NP = BPP, there is no polynomial-time $d(1/3 - \eps)$-approximation algorithm for $d$-bounded-degree Max 2-CSP.  
\end{theorem}

\paragraph{Independent Set in Claw-Free Graphs.}
Leveraging the connection between bounded-degree Max 2-CSP and Maximum Independent Set in claw-free graphs~\cite{DvorakFRR23} discussed above, we arrive at a $\approx k/4$ hardness for the latter, assuming the Unique Games Conjecture.

\begin{theorem} \label{thm:ug-claw-free-main}
Assuming the UGC, for any $\eps \in (0, 1/4)$, there exists $k_0 \in \N$ such that the following holds for every positive integer $k \geq k_0$: Unless NP = BPP, there is no polynomial-time $k(1/4 - \eps)$-approximation algorithm for Maximum Independent Set on $k$-claw-free graphs.  
\end{theorem}

Again, this is the first $\Omega(k)$ hardness for the problem with an explicit constant. Furthermore, this is within a factor of $2$ (as $k \to \infty$) of the aforementioned polynomial-time approximation algorithms~\cite{Neuwohner21,ThieryW23} and within a factor of $4/3 + o(1)$ of the quasi-polynomial time approximation algorithm~\cite{CyganGM13}.

For NP-hardness result, we get a slightly weaker factor that is $\approx k/5.829$ instead:

\begin{theorem} \label{thm:np-claw-free-main}
For any $\eps \in \left(0, \frac{1}{3 + 2\sqrt{2}}\right)$, there exists $k_0 \in \N$ such that the following holds for every positive integer $k \geq k_0$: Unless NP = BPP, there is no polynomial-time $k\left(\frac{1}{3 + 2\sqrt{2}} - \eps\right)$-approximation algorithm for Maximum Independent Set on $k$-claw-free graphs.  
\end{theorem}

\subsection{Technical Overview}
We now briefly (and informally) discuss our techniques. Perhaps surprisingly, we use the same strategy as in previous work: sparsify a dense(er) 2-CSP instance by randomly sampling its constraints. This strategy--originated in~\cite{Trevisan01}--has been used in many subsequent papers on the topic (e.g. \cite{GoldreichS06,Laekhanukit14,DinitzKR16,Manurangsi19}). As we will elaborate more below, the main ``twist'' in our work is that, instead of starting with Max 2-CSP hardness with a gap roughly similar to the desired gap after sampling, we start with Max 2-CSP hardness with a much larger gap. 

To discuss this in more detail, let us first recall the standard subsampling procedure. We start with a 2-CSP instance $\Pi$ and produces $\Pi'$ as follows: (i) keep each edge in $\Pi$ with probability $p$ and (ii) remove edges until the every vertex has degree at most $d$. For simplicity of presentation, let us assume that the constraint graph in $\Pi$ is $\td$-regular. In this case, we can let $p = d/\td$. The completeness of the reduction is obvious: if $\Pi$ is fully satisfiable\footnote{Again, this is for simplicity; in the actual reduction, we only have almost satisfiability.}, then $\Pi'$ is also fully satisfiable.

The main challenge is in analyzing the soundness. Again, suppose for simplicity that we did not have to apply step (ii). Suppose that any assignment satisfies less than $\gamma$ fraction of constraints in $\Pi$. What can we say about $\Pi'$?

A standard soundness argument here is to use a concentration bound to show that, for some $\gamma' > \gamma$ and any fixed assignment $\psi$, the probability that $\psi$ satisfies more than $\gamma'$ fraction of constraints in $\Pi'$ is at most $q$. Then, using a union bound over all assignments, one arrive at a conclusion that no assignment satisfies more than $\gamma'$ fraction of constraints in $\Pi'$. This gives a gap of $\gamma'$. Recall that we want $\gamma' = \Omega(1)/d$. 
 Note also that there are $R^n$ assignments, where $R$ denote the alphabet size and $n$ denote the number of variables. Therefore, we need $q \ll R^{-n}$ for this argument to work. Meanwhile, when $\gamma = (\gamma')^{\omega(1)}$ the multiplicative Chernoff bound gives
 \begin{align*}
 q \leq O\left(\frac{\gamma'}{\gamma}\right)^{\gamma' d n / 2} = (1/\gamma)^{-\gamma' dn/2 \cdot (1-o(1))}.
 \end{align*}
Comparing this with the required $q \ll R^{-n}$, it suffices for us to take $\gamma' = \frac{2}{d} \cdot \log_{1/\gamma}R \cdot (1 + o(1))$. Putting it differently, if we start with a hardness for Max 2-CSP with a gap of $1/\gamma = R^\nu$, then we end up with a gap of $\frac{d}{2} \cdot \nu \cdot (1 - o(1))$. Under the UGC, we show that such a hardness can be proved for $\nu = 1 - o(1)$. (See \Cref{sec:ug-2csp}.) This immediately yields \Cref{thm:ug-2csp-main}.

 It is crucial to point out that we require $\gamma \ll \gamma'$ as otherwise, if $\gamma' = \Theta(\gamma)$, the bound would only be $\exp\left(-\Omega\left(\gamma' d n\right)\right)$ which would require us to take $\gamma' = \Omega\left(\frac{\log R}{d}\right)$. This $\log R$ factor is essentially what differentiates us from previous work on similar topics, such as \cite{Laekhanukit14}.

\paragraph{Optimizing Parameters for NP-hardness.} For NP-hardness, we start with the NP-hardness result of \cite{Chan16} where $\gamma = R^{1/2 - o(1)}$ or $\nu = 1/2 - o(1)$. If we were to plug this into the above argument directly, we would get a gap of only $\frac{d}{4} \cdot (1 - o(1))$. We are able to get a better gap of $\frac{d}{3} \cdot (1 - o(1))$ by observing that the instance of \cite{Chan16} is bipartite and has RHS alphabet of size only $\sqrt{R}$. This allows us to use a union bound on only $R^{3n/4}$ assignments (instead of $R^n$), which improves the inapproximability ratio as claimed.

\paragraph{Independent Set on $k$-Claw-Free Graphs.} For UGC-hardness of of Maximum Independent Set on $k$-claw-free graphs, we can combine the UGC-hardness for Max 2-CSP with bounded degree (\Cref{thm:ug-2csp-main}) together with the aforementioned connection from \cite{DvorakFRR23}, which immediately yields \Cref{thm:ug-claw-free-main}. As for NP-hardness, the same strategy only gives us $\frac{k}{6} \cdot (1 - o(1))$ inapproximability factor. To improve this, we observe that the observation in \cite{DvorakFRR23} can be further refined when the graph is bipartite and the maximum degree on each sides are different (\Cref{lem:red-fglss}). By balancing these degree parameters (namely letting the LHS degree being $\approx \sqrt{2}$ times that of the RHS), we arrive at the claimed $\frac{k}{3 + 2\sqrt{2}} \cdot (1 - o(1))$ hardness factor.

\subsection{Other Related Work}

Maximum Independent Set on $k$-claw-free graphs is closely related to many other important problems in literature. For example, it contains Maximum Independent Set on bounded-degree graphs (where the maximum degree is at most $k$) as a special case. It turns out that the latter is easier: a $\tilde{O}(k/\log^2 k)$-approximation algorithm~\cite{BansalGG18} is known and this is essentially tight~\cite{AustrinKS11}. Another closely related problem is the $k$-Set Packing problem, in which we are given sets of size at most $k$ and would like to pick as many disjoint sets as possible. It is simple to see that the problem is equivalent to finding an independent set in the ``conflict graph''--where each set becomes a vertex and two vertices are linked if and only if the sets intersect--and that this conflict graph is $(k + 1)$-claw-free. Thus, all aforementioned approximation algorithms for Maximum Independent Set on claw-free graphs immediately apply to $k$-Set Packing. However, the latter can also be less challenging: the aforementioned quasi-polynomial time algorithm of Cygan et al.~\cite{CyganGM13} for the former can be sped up to run in polynomial time while acheiving a similar approximation ratio~\cite{Cygan13,SviridenkoW13}. Meanwhile, the best known hardness of approximation for the problem remains an NP-hardness with approximation factor $\Omega(k/\log k)$ due to Hazan et al.~\cite{HazanSS06}.

\section{Preliminaries}

We use $\indp(G)$ to denote the size of the maximum independent set in the graph $G$.

\subsection{Concentration Inequalities}

We recall the standard multiplicative Chernoff bound:
\begin{theorem}[Multiplicative Chernoff Bound] \label{thm:mult-chernoff}
Let $X_1, \dots, X_m$ be i.i.d. Bernoulli random variable with mean at most $\mu$, and let $S = X_1 + \cdots + X_m$. Then, for any $\theta > \mu m$, we have
\begin{align*}
\Pr\left[S > \theta\right] < 
\exp(\theta - \mu m) \left(\frac{\mu m}{\theta}\right)^{\theta}
\end{align*}
\end{theorem}

For $\tau > 0$ and $x \in \R$, let $\clip_{\tau}(x) := \min\{x, \tau\}$. We will need the following lemma for the purpose of analyzing edge removal to bound the maximum degree of a random subgraph.

\begin{lemma} \label{lem:clip}
Let $X_1, \dots, X_m$ be i.i.d. Bernoulli random variable with mean at most $\mu$, and let $S = X_1 + \cdots + X_m$. Then, for any integer $\tau > \mu m$, we have
\begin{align*}
\E\left[S - \clip_{\tau}(S)\right] \leq \left(\frac{\mu m}{\tau - \mu m}\right)^2.
\end{align*}
\end{lemma}

\begin{proof}
For any $i \in [m]$, we have $\Var[X_1 + \cdots + X_{i - 1}] = \Var[X_1] + \cdots \Var[X_{i-1}] \leq \mu(i - 1)$.

We can rewrite the LHS as
\begin{align*}
\E[S - \clip_{\tau}(S)] 
&= \E\left[\sum_{i \in [m]} X_i \cdot \ind[X_1 + \cdots X_{i - 1} \geq \tau]\right] \\
&= \sum_{i \in [m]} \E[X_i] \cdot \Pr[X_1 + \cdots X_{i - 1} \geq \tau] \\
&\leq \sum_{i \in [m]} \mu \cdot \frac{\mu(i - 1)}{(\tau - \mu(i - 1))^2} &\leq \left(\frac{\mu m}{\tau - \mu m}\right)^2,
\end{align*}
where in the first inequality we use Chebyshev's inequality.
\end{proof}

\subsection{Constraint Satisfaction Problems}

Formal definitions of a 2-CSP instance and its assignment are given below.

\begin{definition}
A 2-CSP instance $\Pi$ consists of:
\begin{itemize}
\item \emph{Constraint graph} $G = (V, E)$.
\item \emph{Alphabet} $\Sigma_v$ for all $v \in V$.
\item For each $e = (u, v) \in E$, a \emph{constraint} $R_e \subseteq \Sigma_u \times \Sigma_v$.
\end{itemize}

An \emph{assignment} is a tuple $(\psi_v)_{v \in V}$ such that $\psi_v \in \Sigma_v$. Its value $\val_{\Pi}(\psi)$ is defined as the fraction of edges $e = (u, v) \in E$ such that $(\psi_u, \psi_v) \in E$; such an edge (or constraint) is said to be satisfied. The value of the instance is defined as $\val(\Pi) = \max_{\psi} \val_{\Pi}(\psi)$ where the maximum is over all assignments $\psi$.
\end{definition}

Additionally, we use the following terminologies for CSPs:
\begin{itemize}
\item A 2-CSP instance is \emph{$d$-bounded-degree} if the every vertex in the constraint graph $G$ has degree at most $d$.
\item The \emph{alphabet size} of a 2-CSP instance is $\max_{v \in V} |\Sigma_v|$. 
\item A 2-CSP instance is \emph{bipartite} if the constraint graph $G = (A, B, E)$ is a bipartite graph. 
\item A bipartite 2-CSP instance is \emph{$(d_1, d_2)$-biregular} if every left-hand side vertex (in $A$) has degree $d_1$ and every right-hand side vertex (in $B$) has degree $d_2$.
\item A bipartite 2-CSP instance is \emph{$(d_1, d_2)$-bounded-degree} if every left-hand side vertex (in $A$) has degree at most $d_1$ and every right-hand side vertex (in $B$) has degree at most $d_2$.
\item The left (resp. right) alphabet size of a bipartite 2-CSP instance is $\max_{a \in A} |\Sigma_a|$ (resp. $\max_{b \in B} |\Sigma_b|$).
\end{itemize}

\subsection{Hardness of 2-CSP in terms of Alphabet Size}

As discussed in the introduction, we need hardness of almost-perfect completeness 2-CSP with a gap that is polynomial in the alphabet size. For NP-hardness, the best known result is due to \cite{Chan16}, which has a gap of $R^{1/2 - o(1)}$:

\begin{theorem}[\cite{Chan16}] \label{thm:np-2csp-unbounded-deg}
For any $\zeta > 0$ and sufficiently large $R \in \N$ such that $\sqrt{R}$ is a prime number, there exists $d_1, d_2 \in \N$ such that it is NP-hard, given a bipartite $(d_1, d_2)$-biregular\footnote{Note that biregularity follows immediately if we intiate the reduction of \cite{Chan16} with a biregular Label Cover.} 2-CSP $\Pi$ with left alphabet size $R$ and right alphabet size $\sqrt{R}$, to distinguish between the following two cases:
\begin{itemize}
\item (Yes Case) $\val(\Pi) \geq 1 - \zeta$,
\item (No Case) $\val(\Pi) \leq O\left(\frac{\log R}{\sqrt{R}}\right)$.
\end{itemize}
\end{theorem}

For UGC-hardness, a standard proof technique by Khot, Kindler, Mossel, and O'Donnell~\cite{KhotKMO07} yields a hardness of factor $R^{1 - o(1)}$, as stated below. Since we are not aware\footnote{While Kindler et al.~\cite{KindlerKT16} proved a UGC-hardness result with a similar factor, their result does not satisfy almost-perfect completeness, making it unsuitable for our purpose.} of such a result fully written down in literature, we provide its proof in \Cref{sec:ug-2csp} for completeness.

\begin{theorem} \label{thm:ug-2csp-unbounded-deg}
Assuming the Unique Games Conjecture, for any $\zeta > 0$ and sufficiently large $R \in \N$, there exists $d_1, d_2 \in \N$ such that it is NP-hard, given a bipartite $(d_1, d_2)$-biregular 2-CSP $\Pi$ with alphabet size $R$, to distinguish between the following two cases:
\begin{itemize}
\item (Yes Case) $\val(\Pi) \geq 1 - \zeta$,
\item (No Case) $\val(\Pi) \leq O\left(\frac{\log^2 R}{R}\right)$.
\end{itemize}
\end{theorem}

\section{Hardness of Bounded-Degree 2-CSP}

In this section, we present our main reduction and prove \Cref{thm:ug-2csp-main} and \Cref{thm:np-2csp-main}.

\subsection{Adjusting the Degrees}

As alluded to in the introduction, it will be useful to have a flexible control of the degrees of the two sides of the constraint graph. This can be easily done by copying the vertices on each side, as formalized below.

\begin{lemma} \label{lem:copy}
For any $d_1, d_2, c_1, c_2 \in \N$, there is a polynomial-time reduction from a bipartite $(d_1, d_2)$-biregular 2-CSP $\Pi$ to a bipartite $(c_2 d_1 d_2, c_1 d_1 d_2)$-biregular 2-CSP $\Pi'$ such that $\val(\Pi') = \val(\Pi)$. Moreover, the reduction preserves the left and right alphabet sizes.
\end{lemma}

\begin{proof}
Let the original 2-CSP instance be $\Pi = (G = (A, B, E), (\Sigma_v)_{v \in A \cup B}, (R_e)_{e \in E})$ where $G$ is $(d_1, d_2)$-biregular. We define $\Pi = (G' = (A', B', E'), (\Sigma_{v'})_{v' \in A' \cup B'}, (R_{e'})_{e' \in E'})$ where
\begin{itemize}
\item $A' = A \times [d_1] \times [c_1]$,
\item $B' = B \times [d_2] \times [c_2]$,
\item $E' = \{((a, i_1, j_1), (b, i_2, j_2)) \mid (a, b) \in E, i_1 \in[d_1], i_2\in[d_2], j_1\in[c_1], j_2\in[c_2]\}$,
\item $\Sigma'_{(v, i, j)} = \Sigma_v$ for all $(v, i, j) \in A' \cup B'$.
\item $R_{((a, i_1, j_1), (b, i_2, j_2))} = R_{(a, b)}$ for all $((a, i_1, j_1), (b, i_2, j_2)) \in E'$.
\end{itemize}
To see that $\val(\Pi') \geq \val(\Pi)$, let $\psi$ denote the assignment of $\Pi$ with $\val_{\Pi}(\psi) = \val(\Pi)$. Define an assignment $\psi'$ of $\Pi'$ such that $\psi'_{(v, i, j)} := \psi_v$. it is simple to check that $\val_{\Pi'}(\psi') = \val_{\Pi}(\psi) = \val(\Pi)$.

On the other hand, to see that $\val(\Pi') \leq \val(\Pi)$, notice that $E'$ is can be partitioned into $E_{(i_1, j_1, i_2, j_2)} := \{((a, i_1, j_1), (b, i_2, j_2)) \mid (a, b) \in E\}$ where $i_1 \in[d_1], i_2\in[d_2], j_1\in[c_1], j_2\in[c_2]$. Thus, since any assignment satisfies at most $\val(\Pi)$ fraction of $E_{(i_1, j_1, i_2, j_2)}$, we can conclude that any assignment also satisfies at most $\val(\Pi)$ fraction of $E$.
\end{proof}

\subsection{Main Reduction: Degree Reduction via Subsampling}

We are now ready to state our main reduction and its properties. For readers interested in only the UGC-hardness results, it suffices to think of just the case where the degree bounds $d_A, d_B$ are equal, the alphabet sizes are equal (i.e. $t = 1$) and $\nu = 1 - o(1)$ in the theorem statement below.

\begin{theorem} \label{thm:subsampling-main}
For any $t, \delta, \nu \in (0, 1]$ such that $\delta < \nu$, any positive integer $C$, and any sufficiently large positive integers $d_A, d_B \geq d_0(\delta, \nu)$ and $R \geq R_0(\delta, \nu, t, d_A, d_B)$%
, the following holds: there is a randomized polynomial-time reduction from a bipartite $(d_A C, d_B C)$-biregular 2-CSP $\Pi'$ with left alphabet size at most $R$ and right alphabet size at most $R^t$ to a $(d_A, d_B)$-bounded-degree 2-CSP $\Pi''$ such that, with probability $2/3$, we have
\begin{itemize}
\item (Completeness) $\val(\Pi'') \geq \val(\Pi') - \delta$, and,
\item (Soundness) If $\val(\Pi'') \leq \frac{1}{R^\nu}$, then $\val(\Pi') \leq \frac{1}{\nu - \delta}\left(\frac{1}{d_A} + \frac{t}{d_B}\right)$.
\end{itemize}
\end{theorem}

\begin{proof}
Let $\Pi' = (G' = (A', B', E'), (\Sigma_{v'})_{v' \in A' \cup B'}, (R_e)_{e \in E'})$ denote the original instance.
We select the parameters as follows:
\begin{itemize}
\item $\lambda := 0.001 \min\{\delta, \nu\}$,
\item $p := \frac{1 - \lambda}{C}$,
\item $d_0 = \frac{10000}{\lambda^3}$,
\item $\chi := \frac{1}{\nu - 2\lambda}\left(\frac{1}{d_A} + \frac{t}{d_B}\right)$,
\item $R_0 = \max\left\{\left(\frac{e}{\chi}\right)^{1/\lambda}, 100^{1/\left(\frac{1}{d_A} + \frac{t}{d_B} - (\nu - \lambda)\cdot \chi\right)} \right\}$.
\item $n_E := d_A |A'|$. 
\end{itemize}

We construct the instance $\Pi''$ as follows:
\begin{enumerate}
\item First, let $E_1 \subseteq E'$ be a subset of edges where each edge in $E'$ is kept with probability $p$. 
\item Let $E'' = E_1$ and $G'' = (A', B', E'')$. 
\item For all $a \in A'$: If $\deg_{G''}(a) > d_A$, remove (arbitrary) $d_A - \deg_{G''}(a)$ edges adjacent to $a$ from $E''$.
\item For all $b \in B'$: If $\deg_{G''}(b) > d_B$, remove (arbitrary) $d_B - \deg_{G''}(b)$ edges adjacent to $b$ from $E''$.
\item Let $\Pi''$ be $(G' = (A', B', E''), (\Sigma_{v'})_{v' \in A' \cup B'}, (R_e)_{e \in E''})$.
\end{enumerate}

It is obvious by the construction that the instance is $(d_A, d_B)$-bounded degree. Before we prove the completeness and soundness of the reduction, let us briefly give probabilistic bounds on the sizes of $|E_1|$ and $|E_1 \setminus E''|$ that will be useful in both cases.

Let $G_1$ denote $(A', B', E_1)$; furthermore, let $X_e$ denote the indicator variable whether the edge $e$ is included in $E_1$.
Let $\cE_1$ denote the event that $|E_1| \in [(1 - 2\lambda) n_E, n_E]$.
First, we have $\E[|E_1|] = p|E'| = (1 - \lambda) n_E$. Meanwhile, $\Var[|E_1|] = p(1 - p)|E'| \leq p|E'| \leq n_E$. As a result, by Chebyshev's inequality, we have
\begin{align*}
\Pr[\neg \cE_1] \leq \frac{n_E}{\lambda^2 n_E^2} \leq \frac{1}{\lambda^2 n_E} \leq 0.01,
\end{align*}
where the last inequality is due to our choice of on $d_0$.

Let the event $\cE_2$ denote the event that $|E_1 \setminus E''| < \lambda |E_1|$. We have
\begin{align*}
\E[|E_1 \setminus E''|] \leq \sum_{a \in A'} \E\left[\deg_{G_1}(a) - \clip_{d_A}(\deg_{G_1}(a))\right] + \sum_{b \in B'} \E\left[\deg_{G_1}(b) - \clip_{d_B}(\deg_{G_1}(b))\right].
\end{align*}
Observe that $\deg_{G_1}(a)$ (resp. $\deg_{G_1}(b)$) is a sum of $d_A C$ (resp. $d_B C$) i.i.d. random variables with mean $p$. As such, we may apply \Cref{lem:clip} to arrive at
\begin{align*}
\E[|E_1 \setminus E''|]
&\leq |A'| \cdot \frac{(d_A C p)^2}{(d_A - d_A C p)^2} + |B'| \cdot \frac{(d_B C p)^2}{(d_B - d_B C p)^2}
\leq (|A'| + |B'|) \cdot \frac{1}{\lambda^2} 
\leq \frac{2n_E}{d_0 \lambda^2}
\leq 0.01 \lambda |E'|,
\end{align*}
where the last inequality follows from our choice of $d_0$. By Markov's inequality, we thus have
\begin{align*}
\Pr[\neg \cE_2] \leq 0.01.
\end{align*}

\paragraph{Completeness.} Henceforth, for any assignment $\psi$, we use the notation $E'(\psi)$ (resp. $E_1(\psi)$, $E''(\psi)$) to denote the set of edges in $E'$ (resp. $E_1, E''$) satisfied by $\psi$.

Let $\psi^*$ be such that $\val_{\Pi'}(\psi^*) = \val(\Pi')$. Let $\cE_3$ denote the event $E_1(\psi^*) \geq (\val(\Pi') - 2 \lambda) \cdot n_E$. Notice that $E_1(\psi^*)$ is exactly a subset of $E'(\psi^*)$ where each satisfied edge is included with probability $p$. As a result, we have $\E[|E_1(\psi^*)|] = p \cdot |E'(\psi^*)| = p \cdot |E| \cdot \val(\Pi') \geq n_E \cdot (\val(\Pi') - \lambda)$. Meanwhile, we have $\Var[|E_1(\psi^*)|] = p(1 - p) |E'(\psi^*)| \leq n_E \cdot \val(\Pi')$. Thus, by Chebyshev's inequality, we have
\begin{align*}
\Pr[\neg \cE_3] \leq \frac{n_E \cdot \val(\Pi')}{(\lambda n_E)^2} \leq \frac{1}{\lambda^2 n_E} \leq 0.01.
\end{align*}

Thus, by union bound, we have $\Pr[\cE_1 \wedge \cE_2 \wedge \cE_3] \geq 0.97$. When $\cE_1, \cE_2, \cE_3$ all occur, we have
\begin{align*}
\val(\Pi'') \geq \val_{\Pi''}(\psi^*) = \frac{|E''(\psi^*)|}{|E''|} 
\geq \frac{|E_1(\psi^*)| - |E_1 \setminus E''|}{|E_1|} 
&\geq \frac{(\val(\Pi') - 2 \lambda) \cdot n_E - \lambda \cdot n_E}{n_E} \\
&\geq \val(\Pi') - \delta.
\end{align*}

\paragraph{Soundness.}
Assume that $\val(\Pi'') \leq \frac{1}{R^\nu}$ and recall that $\chi = \frac{1}{\nu - 2\lambda}\left( \frac{1}{d_A} + \frac{t}{d_B} \right)$ is slightly smaller than the target soundness.
Consider any assignment $\psi$. Let $\cE_{\psi}$ denote the event that $|E_1(\psi)| \leq \chi \cdot n_E$. Again, notice that $E_1(\psi)$ is exactly a subset of $E'(\psi)$ where each satisfied edge is included with probability $p$. Note also that $|E'(\psi)| \leq \val(\Pi'') \cdot |E'|$, implying that $\E[|E_1(\psi)|] = p \cdot |E'(\psi)| \leq \val(\Pi'') \cdot n_E$. Thus, we may apply \Cref{thm:mult-chernoff} (with $\theta = \chi \cdot n_E$), to arrive at
\begin{align*}
\Pr[\neg \cE_{\psi}] &\leq \exp(\chi \cdot n_E - \val(\Pi'') \cdot n_E) \cdot \left(\frac{\val(\Pi'') \cdot n_E}{\chi \cdot n_E}\right)^{\chi \cdot n_E} 
\leq \left(\frac{e}{\chi \cdot R^{\nu}}\right)^{\chi \cdot n_E} 
\leq R^{-(\nu - \lambda)\cdot \chi \cdot n_E}
\end{align*}
where the third inequality is from our assumption that $R \geq R_0$.

Therefore, by taking the union bound over all (at most $R^{|A'|} \cdot R^{t|B'|}$) assignments $\psi$, we have
\begin{align*}
\Pr\left[\bigvee_{\psi} \neg \cE_\psi\right] \leq R^{|A'|} \cdot R^{t|B'|} \cdot R^{-(\nu - \lambda)\cdot \chi \cdot n_E}
= R^{\frac{n_E}{d_A} + t \cdot \frac{n_E}{d_B}-(\nu - \lambda)\cdot \chi \cdot n_E}
= (R^{n_E})^{\frac{1}{d_A} + \frac{t}{d_B} - (\nu - \lambda)\cdot \chi}
&\leq 0.01.
\end{align*}
Thus, by the union bound, $\cE_1, \cE_2$ and $\cE_{\psi}$ for all assignments $\psi$ occur simultanously with probability at least 0.97. When this is the case, we have
\begin{align*}
\val(\Pi'') 
= \max_{\psi} \frac{|E''(\psi)|}{|E''|}
\leq \frac{|E_1(\psi)|}{|E_1| - |E_1 \setminus E''|}
\leq \frac{\chi \cdot n_E}{(1 - \lambda)n_E} 
\leq \frac{1}{\nu - \delta}\left(\frac{1}{d_A} + \frac{t}{d_B}\right),
\end{align*}
which concludes our proof.
\end{proof}

\subsection{Putting Things Together: Proof of \Cref{thm:ug-2csp-main} and \Cref{thm:np-2csp-main}}

Our main theorems (\Cref{thm:ug-2csp-main,thm:np-2csp-main}) now follow easily from plugging in the reduction above to the existing large-gap hardness results for 2-CSPs (\Cref{thm:ug-2csp-unbounded-deg,thm:np-2csp-unbounded-deg}) with appropriate parameters.

\begin{proof}[Proof of \Cref{thm:ug-2csp-main}]
We will reduce from \Cref{thm:ug-2csp-unbounded-deg} with $\zeta = 0.01\eps$. Let $\Pi$ be any bipartite $(d_1, d_2)$-biregular 2-CSP instance with alphabet size $R$. We first apply \Cref{lem:copy} with $c_1 = c_2 = d$ to arrive at a $(d d_1d_2, d d_1d_2)$-biregular 2-CSP instance $\Pi'$ with the same alphabet size such that $\val(\Pi') = \val(\Pi)$. We then apply the reduction from \Cref{thm:subsampling-main} with $d_A = d_B = d$, $t = 1, \delta = 0.01\eps, \nu = 1 - \delta$ to arrive at a $d$-degree-bounded 2-CSP instance $\Pi''$. When $d$ is sufficiently large (depending on $\eps$ only) and $R$ is sufficiently large (depending on $d, \eps$), with probability 2/3, we have
\begin{itemize}
\item If $\val(\Pi) \geq 1 - \zeta$, then $\val(\Pi'') \geq 1 - \zeta - \delta = 1 - 0.02\eps$.
\item  If $\val(\Pi'') \leq \frac{1}{R^{1 - \delta}}$, then $\val(\Pi'') \leq \frac{1}{1 - 2\delta}\left(\frac{1}{d} + \frac{1}{d}\right) = \frac{1}{1 - 0.02\eps}\left(\frac{2}{d}\right)$.
\end{itemize}
Note that the ratio between the two cases are larger than $d(1/2 - \eps)$. Thus, if there is a polynomial-time $d(1/2 - \eps)$-approximation algorithm for 2-CSP on $d$-bounded-degree graphs, we can distinguish the two cases in randomized polynomial time (with two-sided error). Assuming UGC, from \Cref{thm:ug-2csp-unbounded-deg}, this implies that NP = BPP.
\end{proof}

\begin{proof}[Proof of \Cref{thm:np-2csp-main}]
We will reduce from \Cref{thm:np-2csp-unbounded-deg} with $\zeta = 0.01\eps$. Let $\Pi$ be any bipartite $(d_1, d_2)$-biregular 2-CSP instance with left alphabet size $R$ and right alphabet size $R^{1/2}$. We first apply \Cref{lem:copy} with $c_1 = c_2 = d$ to arrive at a $(d d_1d_2, d d_1d_2)$-biregular 2-CSP instance $\Pi'$ with the same left and right alphabet sizes such that $\val(\Pi') = \val(\Pi)$. We then apply the reduction from \Cref{thm:subsampling-main} with $d_A = d_B = d$, $t = 1/2, \delta = 0.01\eps, \nu = 1/2 - \delta$ to arrive at a $d$-degree-bounded 2-CSP instance $\Pi''$. When $d$ is sufficiently large (depending on $\eps$ only) and $R$ is sufficiently large (depending on $d, \eps$), with probability 2/3, we have
\begin{itemize}
\item If $\val(\Pi) \geq 1 - \zeta$, then $\val(\Pi'') \geq 1 - \zeta - \delta = 1 - 0.02\eps$.
\item  If $\val(\Pi'') \leq \frac{1}{R^{1/2 - \delta}}$, then $\val(\Pi'') \leq \frac{1}{1/2 - 2\delta}\left(\frac{1}{d} + \frac{1/2}{d}\right) = \frac{1}{1 - 0.04\eps}\left(\frac{3}{d}\right)$.
\end{itemize}
Note that the ratio between the two cases are larger than $d(1/3 - \eps)$. Thus, if there is a polynomial-time $d(1/3 - \eps)$-approximation algorithm for 2-CSP on $d$-bounded-degree graphs, we can distinguish the two cases in randomized polynomial time (with two-sided error). From \Cref{thm:np-2csp-unbounded-deg}, this implies that NP = BPP.
\end{proof}

\section{Hardness of Maximum Independent Set in $k$-Claw-Free Graphs}

We next move on to prove hardness of Maximum Independent Set in $k$-claw-free graphs. To do so, let us first recall the reduction from Max 2-CSP with bounded degree from \cite{DvorakFRR23}. As touched on briefly in the introduction, the version we state below is actually more flexible than that in \cite{DvorakFRR23} as it allows the degree bounds of the two sides to be different.

\begin{lemma}[\cite{DvorakFRR23}] \label{lem:red-fglss}
There is a polynomial-time reduction that takes in a $(d_A, d_B)$-bounded degree bipartite 2-CSP instance $\Pi = (G = (A, B, E), (\Sigma_v)_{v \in A \cup B}, (R_e)_{e \in E})$ and produces a $(d_A + d_B)$-claw-free graph $G^* = (V^*, E^*)$ such that $\indp(G^*) = \val(\Pi) \cdot |E|$. 
\end{lemma}

\begin{proof}
The reduction is exactly as the so-called ``FGLSS graph'' \cite{FeigeGLSS96}:
\begin{itemize}
\item For every edge $e = (a, b) \in E$ and every $(\sigma_a, \sigma_b) \in R_e$, create a vertex $(a, b, \sigma_a, \sigma_b)$ in $V^*$.
\item Create an edge between $(a, b, \sigma_a, \sigma_b)$ and $(a', b', \sigma'_{a'}, \sigma'_{b'})$ iff they are inconsistent, i.e. there exists $v' \in \{a, b\} \cap \{a', b'\}$ such that $\sigma_v \ne \sigma'_{v'}$.
\end{itemize}

A standard argument shows that $\indp(G^*) = \val(\Pi) \cdot |E|$. To see that it is $(d_A + d_B)$-claw-free, consider any vertex $(a, b, \sigma_a, \sigma_b) \in V^*$ and $d_A + d_B$ of its neighbors in $G^*$. Let $E_0 \subseteq E$ denote the set of edges (in the constraint graph $G$) adjacent to $a$ or $b$ (or both). By the degree constraint, $|E_0| \leq d_A + d_B - 1$. Thus, by the pigeonhole principle, at least two neighbors of the $d_A + d_B$ neighbors correspond to the same edge in $E_0$; this means that there must be an edge between these two vertices in $G^*$. Thus, the graph $G^*$ is $(d_A + d_B)$-claw-free.
\end{proof}

\Cref{thm:ug-claw-free-main} is now an immediate consequence of \Cref{thm:ug-2csp-main} and \Cref{lem:red-fglss}.

\begin{proof}[Proof of \Cref{thm:ug-claw-free-main}]
Note that the instance produced in \Cref{thm:ug-2csp-main} is a $(d, d)$-bounded-degree instance. Thus, by setting $d = \lfloor k/2 \rfloor$ and plugging it into the reduction in \Cref{lem:red-fglss}, we arrive at the claimed hardness result.
\end{proof}

For \Cref{thm:np-claw-free-main}, we need to work harder to optimize the hardness of approximation factor. Specifically, we set $d_A \approx \sqrt{2} \cdot d_B$, as formalized below.

\begin{proof}[Proof of \Cref{thm:np-claw-free-main}]
Let $q_1, q_2 \in \N$ be integers such that $\left|\frac{q_1}{q_2} - \sqrt{2}\right|$ and $\left|\frac{q_2}{q_1} - \frac{1}{\sqrt{2}}\right| \leq 0.01\eps$. Let $d_A = \lfloor \frac{k q_1}{q_1 + q_2} \rfloor$ and $d_B = \lfloor \frac{k q_2}{q_1 + q_2} \rfloor$. Note that $d_A + d_B \leq k$.

We will reduce from \Cref{thm:np-2csp-unbounded-deg} with $\zeta = 0.01\eps$. Let $\Pi$ be any bipartite $(d_1, d_2)$-biregular 2-CSP instance with left alphabet size $R$ and right alphabet size $R^{1/2}$. We first apply \Cref{lem:copy} with $c_1 = d_A, c_2 = d_B$ to arrive at a $(d_A d_1d_2, d_B d_1d_2)$-biregular 2-CSP instance $\Pi'$ with the same left and right alphabet sizes such that $\val(\Pi') = \val(\Pi)$. We then apply the reduction from \Cref{thm:subsampling-main} with $d_A, d_B$ as specified above, $t = 1/2, \delta = 0.01\eps, \nu = 1/2 - \delta$ to arrive at a $d$-degree-bounded 2-CSP instance $\Pi''$. Finally, we apply reduction in \Cref{lem:red-fglss} on $\Pi''$ to arrive at the graph $G^*$. By \Cref{lem:red-fglss}, $G^*$ is $k$-claw-free. Furthermore, when $k$ is sufficiently large (depending on $\eps$ only) and $R$ is sufficiently large (depending on $k, \eps$), with probability 2/3, we have
\begin{itemize}
\item If $\val(\Pi) \geq 1 - \zeta$, then $\indp(G^*) = |E''| \cdot \val(\Pi'') \geq |E''| \cdot (1 - \zeta - \delta) = |E''| \cdot (1 - 0.02\eps)$.
\item  If $\val(\Pi'') \leq \frac{1}{R^{1/2 - \delta}}$, then 
\begin{align*}
\indp(G^*) = |E''| \cdot \val(\Pi'') &\leq |E''| \cdot  \frac{1}{1/2 - 2\delta}\left(\frac{1}{d_A} + \frac{1/2}{d_B}\right) \\
&= \frac{1}{1/2 - 2\delta} \left(\frac{1}{\frac{kq_1}{q_1 + q_2} - 1} + \frac{1/2}{\frac{kq_2}{q_1 + q_2} - 1}\right) \cdot |E''| \\
&\leq \frac{1}{1/2 - 2\delta} \cdot \frac{1}{1 - \delta}  \left(\frac{q_1 + q_2}{k q_1} + \frac{(q_1 + q_2)/2}{kq_2}\right) \cdot |E''| \\
&\leq \frac{1}{1/2 - 2\delta} \cdot \frac{1}{1 - \delta} \cdot \frac{1}{k}\left(\frac{3}{2} + \frac{q_2}{q_1} + \frac{q_1}{2q_2}\right) \cdot |E''| \\
&\leq \frac{1}{1/2 - 2\delta} \cdot \frac{1}{1 - \delta} \cdot \frac{1}{k}\left(\frac{3}{2} + \sqrt{2} + 2\delta\right) \cdot |E''| \\
&\leq \frac{1}{1 - 4\delta} \cdot \frac{1}{k} \cdot (3 + 2\sqrt{2} + 4\delta) \cdot |E''|,
\end{align*}
where the second inequality holds when we assume that $k$ is sufficiently large and the second-to-last inequality is from our choice of $q_1, q_2$.
\end{itemize}
Note that the ratio between the two cases are larger than $k\left(\frac{1}{3 + 2\sqrt{2}} - \eps\right)$. Thus, if there is a polynomial-time $k\left(\frac{1}{3 + 2\sqrt{2}} - \eps\right)$-approximation algorithm for maximum independent set on $k$-claw-free graphs, we can distinguish the two cases in randomized polynomial time (with two-sided error). Assuming UGC, from \Cref{thm:ug-2csp-unbounded-deg}, this implies that NP = BPP.
\end{proof}

\section{Conclusion and Open Questions}

In this paper, we prove hardness of approximation results for Max 2-CSP with bounded degree. Our UG-hardness is nearly tight as the maximum degree goes to $\infty$. Using this, we also give hardness for Maximum Independent Set on $k$-claw-free graphs whose inapproximation ratio is within a factor of 2 of optimal for any sufficiently large $k$. It remains an intriguing open question to close this latter gap. Furthermore, since our reductions are randomized, it would be interesting to derandomized them. Finally, one of our motivations to study Maximum Independent Set on $k$-claw-free graphs is to understand $k$-Set Packing. However, we are unable to obtain $\Omega(k)$ factor hardness of approximation of the latter using the reductions in this paper. As stated earlier, the best (NP-)hardness of approximation for $k$-Set Packing remains $\Omega(k/\log k)$~\cite{HazanSS06} and it would be interesting to close (or at least decrease) this $O(\log k)$ gap between the upper and lower bounds.

\ifarxiv
\subsection*{Acknowledgement}

This work was initiated at Dagstuhl Seminar 23291 ``Parameterized Approximation: Algorithms and Hardness''. We thank the organizers and participants of the workshop for helpful discussions.
\fi

\bibliographystyle{alpha}
\bibliography{ref}

\newcommand{\etalchar}[1]{$^{#1}$}
\begin{thebibliography}{KKMO07}

\bibitem[AKS11]{AustrinKS11}
Per Austrin, Subhash Khot, and Muli Safra.
\newblock Inapproximability of vertex cover and independent set in bounded
  degree graphs.
\newblock {\em Theory Comput.}, 7(1):27--43, 2011.

\bibitem[Ber00]{Berman00}
Piotr Berman.
\newblock A $d/2$ approximation for maximum weight independent set in $d$-claw
  free graphs.
\newblock {\em Nord. J. Comput.}, 7(3):178--184, 2000.

\bibitem[BGG18]{BansalGG18}
Nikhil Bansal, Anupam Gupta, and Guru Guruganesh.
\newblock On the {Lov{\'{a}}sz} theta function for independent sets in sparse
  graphs.
\newblock {\em {SIAM} J. Comput.}, 47(3):1039--1055, 2018.

\bibitem[BMO{\etalchar{+}}15]{BarakMORRSTVWW15}
Boaz Barak, Ankur Moitra, Ryan O'Donnell, Prasad Raghavendra, Oded Regev, David
  Steurer, Luca Trevisan, Aravindan Vijayaraghavan, David Witmer, and John
  Wright.
\newblock Beating the random assignment on constraint satisfaction problems of
  bounded degree.
\newblock In {\em APPROX}, pages 110--123, 2015.

\bibitem[CCK{\etalchar{+}}20]{ChalermsookCKLM20}
Parinya Chalermsook, Marek Cygan, Guy Kortsarz, Bundit Laekhanukit, Pasin
  Manurangsi, Danupon Nanongkai, and Luca Trevisan.
\newblock From gap-exponential time hypothesis to fixed parameter tractable
  inapproximability: Clique, dominating set, and more.
\newblock {\em {SIAM} J. Comput.}, 49(4):772--810, 2020.

\bibitem[CFK{\etalchar{+}}15]{CyganFKLMPPS15}
Marek Cygan, Fedor~V. Fomin, Lukasz Kowalik, Daniel Lokshtanov, D{\'{a}}niel
  Marx, Marcin Pilipczuk, Michal Pilipczuk, and Saket Saurabh.
\newblock {\em Parameterized Algorithms}.
\newblock Springer, 2015.

\bibitem[CGKS23]{chalermsook2023independent}
Parinya Chalermsook, Ameet Gadekar, Kamyar Khodamoradi, and Joachim Spoerhase.
\newblock Independent set in $k$-claw-free graphs: Conditional
  $\chi$-boundedness and the power of {LP/SDP} relaxations, 2023.

\bibitem[CGM13]{CyganGM13}
Marek Cygan, Fabrizio Grandoni, and Monaldo Mastrolilli.
\newblock How to sell hyperedges: The hypermatching assignment problem.
\newblock In {\em SODA}, pages 342--351, 2013.

\bibitem[Cha16]{Chan16}
Siu~On Chan.
\newblock Approximation resistance from pairwise-independent subgroups.
\newblock {\em J. {ACM}}, 63(3):27:1--27:32, 2016.

\bibitem[CS10]{CHUDNOVSKY2010560}
Maria Chudnovsky and Paul Seymour.
\newblock Claw-free graphs {VI}. colouring.
\newblock {\em Journal of Combinatorial Theory, Series B}, 100(6):560--572,
  2010.

\bibitem[Cyg13]{Cygan13}
Marek Cygan.
\newblock Improved approximation for 3-dimensional matching via bounded
  pathwidth local search.
\newblock In {\em FOCS}, pages 509--518, 2013.

\bibitem[DFRR23]{DvorakFRR23}
Pavel Dvor{\'{a}}k, Andreas~Emil Feldmann, Ashutosh Rai, and Pawel Rzazewski.
\newblock Parameterized inapproximability of independent set in {$H$}-free
  graphs.
\newblock {\em Algorithmica}, 85(4):902--928, 2023.

\bibitem[Din16]{Dinur16}
Irit Dinur.
\newblock Mildly exponential reduction from gap {3SAT} to polynomial-gap
  label-cover.
\newblock {\em Electron. Colloquium Comput. Complex.}, {TR16-128}, 2016.

\bibitem[DKR16]{DinitzKR16}
Michael Dinitz, Guy Kortsarz, and Ran Raz.
\newblock Label cover instances with large girth and the hardness of
  approximating basic $k$-spanner.
\newblock {\em {ACM} Trans. Algorithms}, 12(2):25:1--25:16, 2016.

\bibitem[DM18]{DinurM18}
Irit Dinur and Pasin Manurangsi.
\newblock {ETH}-hardness of approximating 2-{CSP}s and directed steiner
  network.
\newblock In {\em ITCS}, pages 36:1--36:20, 2018.

\bibitem[FGL{\etalchar{+}}96]{FeigeGLSS96}
Uriel Feige, Shafi Goldwasser, L{\'{a}}szl{\'{o}} Lov{\'{a}}sz, Shmuel Safra,
  and Mario Szegedy.
\newblock Interactive proofs and the hardness of approximating cliques.
\newblock {\em J. {ACM}}, 43(2):268--292, 1996.

\bibitem[Fri08]{friedman2008proof}
Joel Friedman.
\newblock {\em A proof of {A}lon's second eigenvalue conjecture and related
  problems}.
\newblock American Mathematical Soc., 2008.

\bibitem[GS06]{GoldreichS06}
Oded Goldreich and Madhu Sudan.
\newblock Locally testable codes and pcps of almost-linear length.
\newblock {\em J. {ACM}}, 53(4):558--655, 2006.

\bibitem[H{\aa}s96]{Hastad96}
Johan H{\aa}stad.
\newblock Clique is hard to approximate within $n^{1-\epsilon}$.
\newblock In {\em FOCS}, pages 627--636, 1996.

\bibitem[H{\aa}s00]{Hastad00}
Johan H{\aa}stad.
\newblock On bounded occurrence constraint satisfaction.
\newblock {\em Inf. Process. Lett.}, 74(1-2):1--6, 2000.

\bibitem[H{\aa}s01]{Hastad01}
Johan H{\aa}stad.
\newblock Some optimal inapproximability results.
\newblock {\em J. {ACM}}, 48(4):798--859, 2001.

\bibitem[HSS06]{HazanSS06}
Elad Hazan, Shmuel Safra, and Oded Schwartz.
\newblock On the complexity of approximating \emph{k}-set packing.
\newblock {\em Comput. Complex.}, 15(1):20--39, 2006.

\bibitem[Kho02]{Khot02}
Subhash Khot.
\newblock On the power of unique 2-prover 1-round games.
\newblock In {\em CCC}, page~25, 2002.

\bibitem[KKMO07]{KhotKMO07}
Subhash Khot, Guy Kindler, Elchanan Mossel, and Ryan O'Donnell.
\newblock Optimal inapproximability results for {MAX-CUT} and other 2-variable
  {CSP}s?
\newblock {\em {SIAM} J. Comput.}, 37(1):319--357, 2007.

\bibitem[KKT16]{KindlerKT16}
Guy Kindler, Alexandra Kolla, and Luca Trevisan.
\newblock Approximation of non-boolean 2{CSP}.
\newblock In {\em SODA}, 2016.

\bibitem[Lae14]{Laekhanukit14}
Bundit Laekhanukit.
\newblock Parameters of two-prover-one-round game and the hardness of
  connectivity problems.
\newblock In {\em SODA}, pages 1626--1643, 2014.

\bibitem[Man19]{Manurangsi19}
Pasin Manurangsi.
\newblock A note on degree vs gap of min-rep label cover and improved
  inapproximability for connectivity problems.
\newblock {\em Inf. Process. Lett.}, 145:24--29, 2019.

\bibitem[Mos10]{Mossel10}
Elchanan Mossel.
\newblock Gaussian bounds for noise correlation of functions.
\newblock {\em Geometric and Functional Analysis}, 19(6):1713--1756, 2010.

\bibitem[MR17]{ManurangsiR17}
Pasin Manurangsi and Prasad Raghavendra.
\newblock A birthday repetition theorem and complexity of approximating dense
  {CSP}s.
\newblock In {\em ICALP}, pages 78:1--78:15, 2017.

\bibitem[MRS21]{ManurangsiRS21}
Pasin Manurangsi, Aviad Rubinstein, and Tselil Schramm.
\newblock The {Strongish Planted Clique Hypothesis} and its consequences.
\newblock In {\em ITCS}, pages 10:1--10:21, 2021.

\bibitem[Neu21]{Neuwohner21}
Meike Neuwohner.
\newblock An improved approximation algorithm for the maximum weight
  independent set problem in d-claw free graphs.
\newblock In {\em STACS}, pages 53:1--53:20, 2021.

\bibitem[Neu23]{Neuwohner-arxiv}
Meike Neuwohner.
\newblock Passing the limits of pure local search for weighted $k$-set packing.
\newblock In {\em SODA}, pages 1090--1137, 2023.

\bibitem[Rag08]{Raghavendra08}
Prasad Raghavendra.
\newblock Optimal algorithms and inapproximability results for every {CSP}?
\newblock In {\em STOC}, pages 245--254, 2008.

\bibitem[SW13]{SviridenkoW13}
Maxim Sviridenko and Justin Ward.
\newblock Large neighborhood local search for the maximum set packing problem.
\newblock In {\em {ICALP}}, pages 792--803, 2013.

\bibitem[Tre01]{Trevisan01}
Luca Trevisan.
\newblock Non-approximability results for optimization problems on bounded
  degree instances.
\newblock In {\em STOC}, pages 453--461, 2001.

\bibitem[TW23]{ThieryW23}
Theophile Thiery and Justin Ward.
\newblock An improved approximation for maximum weighted $k$-set packing.
\newblock In {\em SODA}, pages 1138--1162, 2023.

\bibitem[Zuc07]{Zuckerman07}
David Zuckerman.
\newblock Linear degree extractors and the inapproximability of max clique and
  chromatic number.
\newblock {\em Theory Comput.}, 3(1):103--128, 2007.

\end{thebibliography}

\appendix

\section{UGC-Hardness of 2-CSP with Almost Perfect Completeness}
\label{sec:ug-2csp}

In this section, we prove~\Cref{thm:ug-2csp-unbounded-deg}. 
It follows from a standard technique proving hardness of CSPs assuming the Unique Games Conjecture~\cite{KhotKMO07}. 

\begin{proof}[Proof of \Cref{thm:ug-2csp-unbounded-deg}]

For given $R \in \N$, we will construct a {\em predicate} $P \subseteq [R] \times [R]$, and consider CSP$(P)$ whose instance $\Pi = (G = (V, E), (\Sigma_v)_{v \in V}, (R_{e})_{e \in E})$ must satisfy $\Sigma_v = [R]$ for every $v \in V$ and $R_e = \{ (x, y) \in [R]^2 : (x \oplus t_{e, u}, y \oplus t_{e, v}) \in P \}$ for some $t_{e, u}, t_{e, v} \in [R]$ for every $e = (u, v) \in E$. (For $x, y \in [R]$, we define $x \oplus y$ to be $x + y$ if it is at most $R$ and $x + y - R$ otherwise.)

The standard technique of proving the hardness of CSP$(P)$ due to Khot et al.~\cite{KhotKMO07} shows that it suffices to consider the {\em dictatorship test}. 
It is determined by a distribution $\mu$ supported on $[R] \times [R]$.
For every $L \in \N$, it yields the following test that decides whether a given function $F : [R]^L \to [R]$ is a {\em dictator} or not.

\begin{itemize}
    \item For each $i \in [L]$, sample $(x_{i}, y_{i}) \in [R]^2$ from $\mu$, independently from other $i$'s. 

    \item Accept if $(F(x), F(y)) \in P$. 
\end{itemize}

Note that if $F$ is a dictator function (i.e., $F(x) = x_i$ for some $i \in [L]$), then the above test accepts with probability exactly $\Pr_{(x, y) \sim \mu} [(x, y) \in P] =: c$, known as the completeness of the test. 
Let $s$ be (an upper bound of) the soundness of this test; there exist $\tau > 0$ and $d \in \N$ such that any function $F : [R]^L \to [R]$, which (1) is {\em balanced} (i.e., $|F^{-1}(i)| = R^{L-1}$ for all $i \in [R]$) and (2) has the {\em maximum degree-$d$ influence} (defined in~\Cref{sec:fourier}) at most $\tau$, passes the test with probability at most $s$. Khot et al.~\cite{KhotKMO07} shows that a dictatorship test with some $c$ and $s$ immediately yields the hardness of CSP($P$) with almost the same completeness and soundness. 

\begin{theorem}[\cite{KhotKMO07}] Given $P \subseteq [R]^2$, let $\mu$ be a distribution over $[R]^2$ that yields a dictatorship test with completeness $c$ and soundness $s$. Then, for any $\zeta > 0$, assuming the Unique Games Conjecture, it is NP-hard, given a regular CSP($P$) instance $\Pi$, to distinguish between the following two cases:
\begin{itemize}
    \item (Yes Case) $\val(\Pi) \geq c - \zeta$.
    \item (No Case) $\val(\Pi) \leq s + \zeta$.
\end{itemize}

\end{theorem}

Though the above theorem does not guarantee that a given instance $\Pi$ with the underlying graph $G = (V, E)$ is bipartite, one can easily convert it to a bipartite instance $\Pi'$ by creating two vertex sets $V_1$ and $V_2$ which are disjoint copies of $V$ and replace each constraint $e = (u, v)$ with $R_e \subseteq [R]^2$ with two constraints $(u_1, v_2)$ and $(u_2, v_1)$
with the same $R_e$, where $u_i, v_i$ denote the copy of $u, v$ in $V_i$. The completeness of the new instance is at least the completeness of the original instance, and the soundness of the new instance is at most twice of the soundness of the original instance. 

Therefore, the rest of the section is devoted to constructing $P$ and $\mu$ such that $\mu$ is supported by $P$ (so that $c = 1$) while $s = O(\log^2 R / R)$. 

\subsection{Predicate and Completeness}
Given $R \in \N$, let $t \in \N$ be a parameter to be determined later, and let $H = (V_H, E_H)$ be a $t$-regular graph with $V_H = [R]$ such that the second largest eigenvalue of the normalized adjacency matrix is $O(1/\sqrt{t})$~\cite{friedman2008proof}. The predicate $P \subseteq [R]^2$ is defined such that $(i, j) \in P$ if and only if $( i, j ) \in E_H$. This ensures that $|P| = tR = 2|E_H|$. Then, our distribution $\mu$ is simply the uniform distribution over $P$. By definition, the completeness value is $\Pr_{(x, y) \sim \mu} [(x, y) \in P] = 1$. 

\subsection{Soundness via Fourier analysis}
\label{sec:fourier}
To (formally define and) analyze the soundness of the test, we use the following standard tools from Gaussian bounds for correlated functions from Mossel~\cite{Mossel10}.
We define the correlation between two correlated spaces below.
\begin{definition}
Given a distribution $\mu$ on $\Omega_1 \times \Omega_2$, the correlation $\rho(\Omega_1, \Omega_2; \mu)$ is defined as
\[
\rho(\Omega_1, \Omega_2; \mu) = \sup \left\{ \mathsf{Cov}[f, g] : f : \Omega_1 \rightarrow \mathbb{R}, g : \Omega_2 \rightarrow \mathbb{R}, \mathsf{Var}[f] = \mathsf{Var}[g] = 1 \right\}.
\]
\end{definition}
In our case, $\rho := \rho([R], [R]; \mu)$ is exactly the second largest eigenvalue of the normalized adjacency matrix of $H$, which is $O(1/\sqrt{t})$. 

\begin{definition} [\cite{Mossel10}]
For any function $f : [R]^L \to\R$, the  {\em Efron-Stein} decomposition is given by 
\[
f(y) = \sum_{S \subseteq [L]} f_S (y)
\]
where the functions $f_S$ satisfy 
\begin{itemize}
    \item $f_S$ only depends on $y_S$, the restriction of $y$ to the coordinates of $S$.
    \item For all $S \not\subseteq S'$ and all $z_{S'}$, $\E_y[f_S(y) | y_{S'} = z_{S'}] = 0$. 
\end{itemize}
\end{definition}
Based on the Efron-Stein decomposition, we can define (low-degree) influences of a function. For a function $f : [R]^L \to \R$ and $p \geq 1$, let $\| f \|_p := \E[|f(y)|^p]^{1/p}$.

\begin{definition} [\cite{Mossel10}]\label{def:influence}
For any function $f : [R]^L \to \R$,
its {\em $i$th influence} is defined as 
\[
\Inf_i(f) := \sum_{S : i \in S} \| f_S \|_2^2.
\]
Its {\em $i$th degree-$d$ influence} is defined as 
\[
\Inf^{\leq d}_i(f) := \sum_{S : i \in S, |S| \leq d} \| f_S \|_2^2,
\]
\end{definition}

Given a discrete-valued function $F : [R]^L \to [R]$, for every $i \in [R]$, we let $F_i : [R]^L \to \{ 0, 1 \}$ such that $F_i(x) = 1$ if $F(x) = i$ and $0$ otherwise. We say that $F$ has the maximum degree-$d$ influence at most $\tau$ if $\Inf^{\leq d}_j(F_i) \leq \tau$ for every $i \in [R]$ and $j \in [L]$. 

For $a, b \in [0, 1]$ and $\sigma \in [0, 1]$, let $\Gamma_{\sigma}(a, b) := \Pr[g_1 \leq \Phi^{-1}(a), g_2 \leq \Phi^{-1}(b)]$ where $g_1, g_2$ are $\sigma$-correlated standard Gaussian variables and $\Phi$ denotes the cumulative density function of a standard Gaussian. (E.g., $\Gamma_{\sigma}(a, 1) = a$ for any $\sigma$ and $\Gamma_0(a, b) = ab$.) We crucially use the following invariance principle applied to our dictatorship test. 

\begin{theorem} [\cite{Mossel10}]
For any $\eps > 0$ there exist $d \in \N$ and $\tau > 0$ such that the following is true. Let $f, g : [R]^L \to [0, 1]$. If $\min(\Inf^{\leq d}_i[f], \Inf^{\leq d}_i[g]) \leq \tau$ for every $i \in [L]$, \[
\E_{(x, y) \sim \mu^{\otimes L}}[f(x)g(y)] \leq \Gamma_{\rho}(\E[f(x)], \E[g(y)]) + \eps.
\]
\label{thm:invariance}
\end{theorem}

We will use the following upper bound on $\Gamma_{\rho}(\alpha, \alpha)$. 
\begin{lemma} [Corollary 3 of \cite{KhotKMO07}] For any $R \geq 1$ and $\rho \in (0, 1/20)$, 
\[
\Gamma_{\rho}(1/R, 1/R) \leq (1/R)^{1 + \frac{1 - \rho}{1 + \rho}} \leq 
(1/R)^{2 - 2\rho}.
\]
\end{lemma}
Fix $\eps = 1/R^3$ to get $d$ and $\tau$ from~\Cref{thm:invariance}.
Then, for any $F : [R]^L \to [R]$ that is balanced (i.e., $|F^{-1}(i)| = R^{L-1}$ for every $i \in [R]$) and has the maximum degree-$d$ influence at most $\tau$, the probability that the dictatorship test accepts is 
\[
\sum_{(i, j) \in P}
(\E_{(x, y) \sim \mu^{\otimes L}} [F_i(x) F_j(y)] + \eps)
\leq tR \cdot (\Gamma_{\rho}(1/R, 1/R) + \eps)
\leq t \cdot (1/R)^{1 - 2\rho} + 1/R. 
\]
Recalling $\rho = O(1 / \sqrt{t})$ and setting $t = \Theta(\log^2 R)$ ensure that $t (1/R)^{1-2\rho} \leq O(\log^2  R / R)$, so the soundness is at most $O(\log^2 R / R)$.
\end{proof}

\section{Approximation Algorithm}
\label{sec:approx-algo}

In this section, we give a $(\frac{d+1}{2})$-approximation full algorithm for any $d$-bounded-degree 2-CSP. Before we proceed to the algorithm, let us note that in the case where the instance is fully satisfiable, there is a simple algorithm: Just take any spanning forest of the constraint graph and then use a dynamic programming algorithm to find an assignment that satisfies all the edges in the spanning forest! This algorithm does not work in the general case since it is possible that this spanning forest has a small value. To overcome this, below we sample the spanning forest from an appropriate distribution, allowing us to maintain the same approximation ratio. 

\begin{theorem}
There is a polynomial-time $(\frac{d+1}{2})$-approximation algorithm for every 
$d$-bounded-degree $2$-CSP.
\end{theorem}
\begin{proof}
Let $\Pi = (G = (V, E), (\Sigma_v)_{v \in V}, (R_e)_{e \in E})$ be an instance of $2$-CSP where the maximum degree of $G$ is at most $d$. 

Let $x \in \R^E$ be such that $x_e = 2/(d+1)$ for every $e \in E$. We claim that $x$ is inside the graphic matroid polytope induced by $G$; for any $S \subseteq V$, if $|S| \in [2, d+1]$, $x(E(S)) \leq \frac{|S|(|S|-1)}{2}(\frac{2}{d+1}) \leq |S| - 1$, and if $|S| > d + 1$, $x(E(S)) \leq \frac{d|S|}{2}\cdot(\frac{2}{d+1}) \leq |S|-1$.

Therefore, there exists a distribution $\mathcal{T}$ of forests such that for a random forest $T \sim \mathcal{T}$, for every $e \in E$, $\Pr[e \in T] \geq \frac{2}{d + 1}$. Then, using dynamic programming, one can optimally solve the subinstance of $\Pi$ induced by $T$; for each connected component $T'$ of $T$ (which is a tree), root it at an arbitrary vertex, and for each node $v \in T'$ and $\sigma \in \Sigma_v$, let $A(v, \sigma)$ be the the optimal value of the CSP induced by the subtree of $T'$ rooted at $v$ when the variable $v$ is assigned label $\sigma$. One can compute $A(v, \sigma)$ in a bottom-up fashion using dynamic programming.

Since $\Pr[e \in T] \geq \frac{2}{d + 1}$ for every $e \in E$, the expected optimal number of satisfied constraints of the CSP instance induced by $T$ is at least $\frac{2}{d + 1} \cdot |E| \cdot \val(\Pi)$. Therefore, returning an optimal assignment for a random $T$ yields a $\left(\frac{d+1}{2}\right)$-approximation in expectation. It can be easily derandomized since the integrality of the matroid polytope ensures that the support of $\mathcal{T}$ is polynomial-sized. 
\end{proof}

\section{Parameterized Hardness of Approximation}
\label{sec:fpt}

In this section, we briefly discussed the parameterized hardness of approximation for the Maximum Independent Set in $k$-claw-free graphs.
Recall that an algorithm is said to be fixed parameter tractable (FPT) w.r.t. to parameter $q$ if it runs in time $f(q) \cdot n^{O(1)}$ where $f$ can be any function and $n$ is the input size. We refer interested readers to \cite{CyganFKLMPPS15} for more background on the topic.

Similar to \cite{DvorakFRR23}, we let the parameter be $q = k + \indp(G)$. For this parameter, \cite{DvorakFRR23} showed (by reducing from parameterized hardness of Max 2-CSP in~\cite{ManurangsiRS21}) that, assuming the (less standard) Strongish Planted Clique Hypothesis, no FPT algorithm achieves $o(k)$-approximation. Note that this is incomparable to hardness presented in the main body of our paper (\Cref{thm:np-claw-free-main,thm:ug-claw-free-main}), as such a parameterized hardness result does not rule out e.g. $n^{O(k)}$-time algorithm. (Our main results rule out such algorithms since $k$ there are simply absolute constants.)

Meanwhile, under the (arguably more standard) Gap-ETH assumption\footnote{Gap Exponential Time Hypothesis (Gap-ETH)~\cite{Dinur16,ManurangsiR17} states that no $2^{o(n)}$-time algorithm can distinguish between a fully satisfiable 3-SAT instance and one which is not even $(1 - \eps)$-satisfiable for some constant $\eps > 0$.}, \cite{DvorakFRR23} only show (via a reduction from parameterized hardness of Max 2-CSP in~\cite{DinurM18}) that no FPT algorithm achieves $o\left(\frac{k}{2^{(\log k)^{1/2 + o(1)}}}\right)$-approximation. Our result here is an improvement of this factor to $o(k)$:
\begin{theorem} \label{thm:param-indp-claw-free}
Assuming Gap-ETH, there is no $f(k) \cdot n^{O(1)}$-time $o(k)$-approximation algorithm for Maximum Independent Set on $k$-claw-free graphs even when the maximum independent set has size at most $k$.
\end{theorem}

To prove this theorem, we need the following additional notations for 2-CSPs:
\begin{itemize}
\item For a 2-CSP instance $\Pi = (G = (V, E), (\Sigma_v)_{v \in V}, (R_e)_{e \in E})$, a \emph{partial assignment} is a tuple $(\psi_v)_{v \in V}$ such that $\psi_v \in \Sigma_v \cup \{\perp\}$. Its \emph{size} is defined as  $|\{v : \psi_v \neq \perp\}|$.
\item We say that a partial assignment $\psi$ is consistent if, for all $e = (u, v) \in V$ such that $\psi_u, \psi_v \ne \perp$, we have $(\psi_u, \psi_v) \in R_e$. 
\item Finally, we define $\cval(\Pi)$ to be the maximum size of any consistent partial assignment.
\end{itemize}

We will use the following hardness result\footnote{In \cite{ChalermsookCKLM20}, this is stated as the hardness of Clique, but this is exactly the same as 2-CSP with $k$ variables.}:
\begin{theorem}[\cite{ChalermsookCKLM20}] \label{thm:k-clique-param}
Assuming Gap-ETH, there is no $f(k) \cdot n^{O(1)}$-time $o(k)$-approximation algorithm for $\cval(\Pi)$ with $k$ variables.
\end{theorem}

Our main ingredient is the following reduction, which is different than that of \cite{DvorakFRR23} and allows us to use $\cval$ instead of $\val$ for (in)approximation purposes.

\begin{lemma}
There is a polynomial-time reduction that takes in a 2-CSP instance $\Pi = (G = (V, E), (\Sigma_v)_{v \in V}, \{R_e\}_{e \in E})$ and produces a graph $G' = (V', E')$ such that $\indp(G') = \cval(\Pi)$. Moreover, if $G$ has degree at most $d$, then $G'$ is $(d + 2)$-claw-free.  
\end{lemma}

\begin{proof}
Let $G'$ be the label-extended graph of $G$. Namely, $V' = \{(v, \sigma_v) \mid v \in V, \sigma_v \in \Sigma_v\}$ and there is an edge between $(u, \sigma_u)$ and $(v, \sigma_v)$ in $E'$ iff $(u, v) \in E$ and $(\sigma_u, \sigma_v) \notin R_e$. The claim $\cval(\Pi) = \indp(V')$ is obvious. To see that the graph $G'$ is $(d+2)$-claw-free, observe that any vertex $(u, \sigma_u)$ is only neighbors to $(v, \sigma_v)$ where $v \in N_G[u]$ (the closed-neighbor of $G$). However, for each fixed $v$, $\{(v, \sigma_v) \mid \sigma_v \in \Sigma_v\}$ forms a clique. Thus, the largest size of claw that is a subgraph of $G'$ is at most $|N_G[u]| \leq d + 1$.
\end{proof}

Plugging in the above lemma to~\Cref{thm:k-clique-param}, we immediately arrive at \Cref{thm:param-indp-claw-free}.

\end{document}